\documentclass[abstract=true, DIV=14]{scrartcl}

\usepackage[noEnd,commentColor=black]{algpseudocodex}
\usepackage{amsmath, mleftright}
\usepackage{amssymb}
\usepackage{amsthm}
\usepackage[mathscr]{euscript}
\usepackage[shortlabels]{enumitem}
\usepackage{graphicx} %
\usepackage{subcaption}
\usepackage{thmtools}
\usepackage{todonotes}
\usepackage{xcolor}
\usepackage{microtype}
\usepackage[T1]{fontenc}
\usepackage{nicefrac,xfrac}

\makeatletter\input{t1cmss.fd}\makeatother
\DeclareFontShape{T1}{cmss}{bx}{n}%
  {<5><6><7><8><9><10->ecsx1000}{}
  
\makeatletter
\def\@fnsymbol#1{\ensuremath{\ifcase#1\or \!\;\or \!\;\or \ddagger\or
   \mathsection\or \mathparagraph\or \|\or **\or \dagger\dagger
   \or \ddagger\ddagger \else\@ctrerr\fi}}
\makeatother

\newcommand{\mybinom}[2]{%
    \left(
    \begin{array}{@{}c@{\,}} #1\\#2 \end{array}
    \right)}

\usepackage[normalem]{ulem}

\usepackage[english]{babel}

\usepackage{tcolorbox}

\usepackage{upgreek}
\usepackage[colorlinks]{hyperref}
\usepackage[capitalize]{cleveref}

\crefname{equation}{eq.}{eqs.} %
\crefname{enumi}{}{} %
\crefname{icase}{case}{cases}
\crefname{ipart}{part}{parts}
\crefname{iprop}{property}{properties}
\crefname{iinv}{invariant}{invariants}

\crefformat{section}{\S\,#2#1#3}
\crefformat{subsection}{\S\,#2#1#3}
\crefrangeformat{section}{\S\S\,#3#1#4--#5#2#6}
\crefmultiformat{section}{\S\S\,#2#1#3}{,~#2#1#3}{,~#2#1#3}{,~#2#1#3}

\usepackage{authblk}
\usepackage[normalem]{ulem}

\usepackage{tikz}
\usetikzlibrary{calc}

\newcommand{\jd}[1]{\textcolor{blue}{\textbf{Justin:} #1}}

\newcommand{\N}{\mathbb{N}}

\newcommand{\fO}{\mathcal{O}}

\newcommand{\F}{\mathcal{F}}
\newcommand{\cF}{\mathcal{F}}

\DeclareMathOperator{\OPT}{\mathsf{OPT}}

\setkomafont{captionlabel}{\bfseries}
\renewcaptionname{english}{\figurename}{Fig.}
\setcapindent{0pt}
\renewcommand{\alpha}{\upalpha}

\newcommand{\cP}{\mathscr{P}}

\newcommand{\cT}{\mathscr{T}}

\newcommand{\cS}{\mathscr{S}}

\newcommand\utimes{\mathbin{\ooalign{$\cup$\cr%
   \hfil\raise0.42ex\hbox{$\scriptscriptstyle\times$}\hfil\cr}}}
\newcommand\bigutimes{\mathop{\ooalign{$\bigcup$\cr%
   \hfil\raise0.36ex\hbox{$\scriptscriptstyle\boldsymbol{\times}$}\hfil\cr}}}

\newcommand{\floor}[1]{\lfloor #1 \rfloor}

\newcommand{\connected}[1]{\def\temp{#1}\ifx\temp\empty\sim\else\overset{#1}{\sim}\fi}

\tikzset{
	point/.style={circle, fill, inner sep=1.5pt},
	smallpoint/.style={point, inner sep=1.2pt},
	tinypoint/.style={point, inner sep=1pt},
	hlbox/.style={fill, {white!90!black}},
	subrect/.style={draw, fill={white!80!cyan}}, %
	msrect/.style=subrect %
}

\newtheorem{theorem}{Theorem}[section]
\newtheorem{lemma}[theorem]{Lemma}

\newtheorem{corollary}[theorem]{Corollary}

\theoremstyle{definition}
\newtheorem{definition}[theorem]{Definition}

\title{Improved space-time tradeoff for TSP via extremal set systems}

\author[1]{Justin Dallant\thanks{$^1$Email: \href{mailto:justin.dallant@tu-dresden.de}{justin.dallant@tu-dresden.de}, \href{mailto:laszlo.kozma@tu-dresden.de}{laszlo.kozma@tu-dresden.de}}}
\author[1]{L\'{a}szl\'{o} Kozma\protect\footnotemark[1]}
\affil[1]{Faculty of Computer Science, TU Dresden, Germany}
\date{}

\begin{document}

\maketitle

\begin{abstract}
The \emph{traveling salesman problem} (TSP) is a cornerstone of combinatorial optimization and has deeply influenced the development of algorithmic techniques in both exact and approximate settings. Yet, improving on the decades-old bounds for solving TSP exactly remains elusive: the dynamic program of Bellman, Held, and Karp from 1962 uses $2^{n+\fO(\log{n})}$ time and space, and the divide-and-conquer approach of Gurevich and Shelah from 1987 uses $4^{n + \fO(\log^2{n})}$ time and polynomial space. A straightforward combination of the two algorithms trades off $T^{n+o(n)}$ time and $S^{n+o(n)}$ space at various points of the curve $ST = 4$. An improvement to this tradeoff when $2 < T < 2\sqrt{2}$ was found by Koivisto and Parviainen (SODA 2010), yielding a minimum of $ST \approx 3.93$. Koivisto and Parviainen show their method to be optimal among a broad class of partial-order-based approaches, and to date, no improvement or alternative method has been found. 

In this paper we give a tradeoff that strictly improves all previous ones for all $2<T<4$, achieving a minimum of $ST < 3.572$. A key ingredient is the construction of sparse set systems (hypergraphs) that admit a large number of maximal chains. The existence of such objects is of independent interest in extremal combinatorics, likely to see further applications. Along the way we disprove a combinatorial conjecture of Johnson, Leader, and Russell from 2013, relating it with the optimality of the previous tradeoff schemes for TSP. Our techniques extend to a broad class of permutation problems over arbitrary semirings, yielding improved space–time tradeoffs in these settings as well.

\end{abstract}

\section{Introduction}
\label{sec1}

The \emph{traveling salesman problem} (TSP) asks, given a set of $n$ cities $S = \{c_1, c_2, \dots, c_n\}$ and distances $d : S^2 \rightarrow \mathbb{R}$, for a tour of minimal length that visits each city in $S$ exactly once. 
More precisely, we seek a permutation $\pi$ of $[n]  = \{1, 2, \dots,n\}$ that minimizes $$d(c_{\pi(n)}, c_{\pi(1)}) + \sum_{i=1}^{n-1} d(c_{\pi(i)}, c_{\pi({i+1})}).$$

TSP is an emblematic problem of computer science that has been investigated since the pioneering era of the field~\cite{Cook2011}. The fascination it holds is likely due to multiple factors such as its clear and intuitive statement, its practical relevance and important special cases (e.g., metric, Euclidean), and the fact that it can be attacked with a variety of algorithmic techniques. The question of the algorithmic complexity of TSP can be traced back to Menger in the 1920s~\cite{schrijver2005history}, and the problem %
continues to inspire research to date, in both exact and approximate settings.  

The 1962 dynamic programming algorithm of Bellman~\cite{Bellman1962} and Held and Karp~\cite{HeldKarp1962} solves TSP by observing that any set of cities that appear contiguously at the beginning of the optimal tour must themselves be visited in an optimal order (as otherwise the entire tour could be improved). The algorithm then computes, for all subsets $S' \subseteq S$ of cities, i.e., all possible \emph{prefix-sets} of the optimal tour, the shortest way of visiting $S'$ from a given start- to a given endpoint. Both the time- and the space requirement\footnote{As common for exponential algorithms, the $\fO^*(\cdot)$ notation suppresses polynomial factors, i.e., $\fO^*(c^n) \subseteq c^{n + \fO(\log{n})}$.} is $\fO^*(2^n)$, dominated by the number of subsets of $S$; %
we recall this algorithm in more detail in \S\,\ref{sec2}.

Improving the running time to $\fO^*(c^n)$ with $c<2$ has been a notorious open problem for well over half a century. For another prototypical NP-hard problem, \textsc{CNF-SAT}, the impossibility of such an improvement is the subject of the Strong Exponential-Time Hypothesis (SETH)~\cite{seth}. Breaking the $2^n$-barrier for TSP would similarly be a major breakthrough. Such improvements have been achieved only in special cases, most notably for the unweighted, undirected case, i.e., for \textsc{Hamiltonicity}: here, an algorithm with running time $\fO^*(1.66^n)$ was found by Björklund~\cite{Bjorklund14}. We refer to the textbook of Fomin and Kratsch~\cite{FominK10} for a broad tour of exponential algorithms, and the survey of Nederlof~\cite{nederlof2026} for a modern treatment. 

\medskip

Generally, exponential space is an even greater obstacle to practicality than exponential time, it is thus often useful to reduce space usage, even at the expense of some increase in running time. 
For TSP, a polynomial space algorithm was found\footnote{The algorithm, as described in~\cite{GurevichShelah1987} is for finding Hamilton cycles; the fact that it easily extends to TSP was noted by Björklund and Husfeldt~\cite{bjorklund2008exact}.} by Gurevich and Shelah~\cite{GurevichShelah1987} in 1987, building on the general divide and conquer approach of Savitch~\cite{savitch1970relationships}.
It works by ``guessing'' the first half of the tour, i.e., a subset $S' \subseteq S$ of $\lfloor n/2 \rfloor$ cities, that appear next to each other in the optimal tour. This amounts to a factor of $\binom{n}{\lfloor n/2 \rfloor} \in \fO(2^n)$ in the running time. Further guessing the cities that begin and end the two halves of the tour (in polynomial time), the algorithm then recursively finds the optimal tours of both parts. The resulting recurrence for the running time is of the form $T(n) \in O^*(2^n) \cdot T(n/2)$, resolving to $T(n) \in 4^n n^{\fO(\log{n})} \subseteq 4^{n + o(n)}$. As there are only $\fO(\log{n})$ recursive levels, each with moderate bookkeeping, the overall space requirement is clearly polynomial.

One may wish to more finely trade off space and time. A straightforward way is to run the above divide and conquer algorithm until depth $i$, then switch over to the dynamic programming (DP) algorithm on the subproblems. This amounts to a runtime of $\fO^*(2^{n(2-1/2^i)}n^i)$ and space usage $\fO^{*}(2^{n/2^i})$; %
see~\cite[\S\,10.1]{FominK10} for a detailed description. 

We call a pair $(S,T)$ a \emph{feasible} space-time tradeoff (for TSP) if every input instance of size $n$ can be solved simultaneously in time $T^{n+o(n)}$ and space $S^{n+o(n)}$. The above algorithms imply that $(2,2)$, $(1,4)$, and various $(S,T)$ pairs with $ST = 4$ are feasible space-time tradeoffs. Exponential improvements to the Bellman-Held-Karp or the Gurevich-Shelah bounds would correspond to tradeoffs $(S,T)$ with $S \leq T < 2$ or $(1,T)$ with $T<4$. The question we study in this paper concerns the entire range between these two extremes. \\
\smallskip

\quad 
\emph{What $(S,T)$ tradeoffs are feasible for $2<T<4$ and what is the minimum of $S\cdot T$?}\\

\smallskip

While $ST=4$ may seem like a fundamental barrier, a 2010 result of Koivisto and Parviainen~\cite{KoivistoParviainen2010}\cite[\S\,10.1]{FominK10} showed that the tradeoff product can be lowered without improving the extremal $T=2$ or $S=1$ cases. More precisely, they obtain, for various values $2 < T < 2\sqrt{2}$ a corresponding feasible $S$ with $ST < 4$, with a minimum tradeoff point of $ST \approx 3.93$. 

At a high level, the Koivisto-Parviainen result works by partitioning the space of possible permutations of $[n]$ (i.e., the set of possible TSP solutions): each part consists of the linear extensions of a partial order (poset) from a certain family. The family of posets found by Koivisto and Parviainen to yield the best result has simple structure: they are of height two and capture the splitting of small groups of cities into two equal parts, one fully preceding the other in the solution. Cities from different groups can arbitrarily intermix. (This latter aspect makes the scheme different from the Gurevich-Shelah divide-and-conquer, where the cities are globally split into two parts that no longer interact.) The algorithm then ``guesses'' the unique poset from the family that the optimal solution extends. Knowing this poset, the number of subproblems that need to be considered in the Bellman-Held-Karp DP is reduced, improving both the time and space bounds. The fact that this scheme yields an overall saving is surprising and somewhat unintuitive -- many natural poset families would, in fact, worsen the tradeoff; even in the favorable class that was identified, improvement happens only at certain group sizes, apparently arising from a precise numerical estimate of the middle binomial coefficient, and as the group size increases or decreases, the effect vanishes. Rigorously optimizing over all possible posets appears out of reach; nonetheless, Koivisto and Parviainen show their choice to be optimal within a broad class. Since 2010 no further improvements or alternative methods have been found for the natural question of the best space-time tradeoff of TSP. 

\medskip

In this paper we develop a general algorithmic framework for space-time tradeoffs that includes the previous two schemes as special cases. It can be viewed as injecting some bias into how likely cities are to be in prefix-sets, more flexibly than the grouping by poset-extensions. It is perhaps most intuitive to describe our approach using randomization, yielding an algorithm that succeeds with high probability; the approach, however, can be made fully deterministic without worsening the bounds. Searching for the optimal tradeoff in this framework turns out to correspond to an extremal problem in combinatorics: finding set systems of a prescribed size that maximize the number of maximal chains. We show that such constructions directly translate to improved space-time tradeoffs for the TSP and other permutation-type problems. 

While natural in hindsight, this extremal problem appears not extensively researched in combinatorics. Questions of a similar flavor were studied, e.g., in a 2013 paper by Johnson, Leader, and Russell~\cite{JLRsetSystems}. %
In fact, the conjectured optimality of a construction by those authors would suggest -- within our broader framework -- the optimality of some poset-based space-time tradeoffs that include the Koivisto-Parviainen scheme.
We disprove their conjecture and give improved set system constructions that in turn lead to improved tradeoffs for the TSP.

Our main results are as follows.

\begin{theorem}
\label{thm1}
For all $2 < T < 4$, there is a value $S$ such that $ST < 4$ and the pair $(S,T)$ is a feasible space-time tradeoff for TSP.
\end{theorem}

\begin{figure}
    \centering
\includegraphics[width=0.85\linewidth]{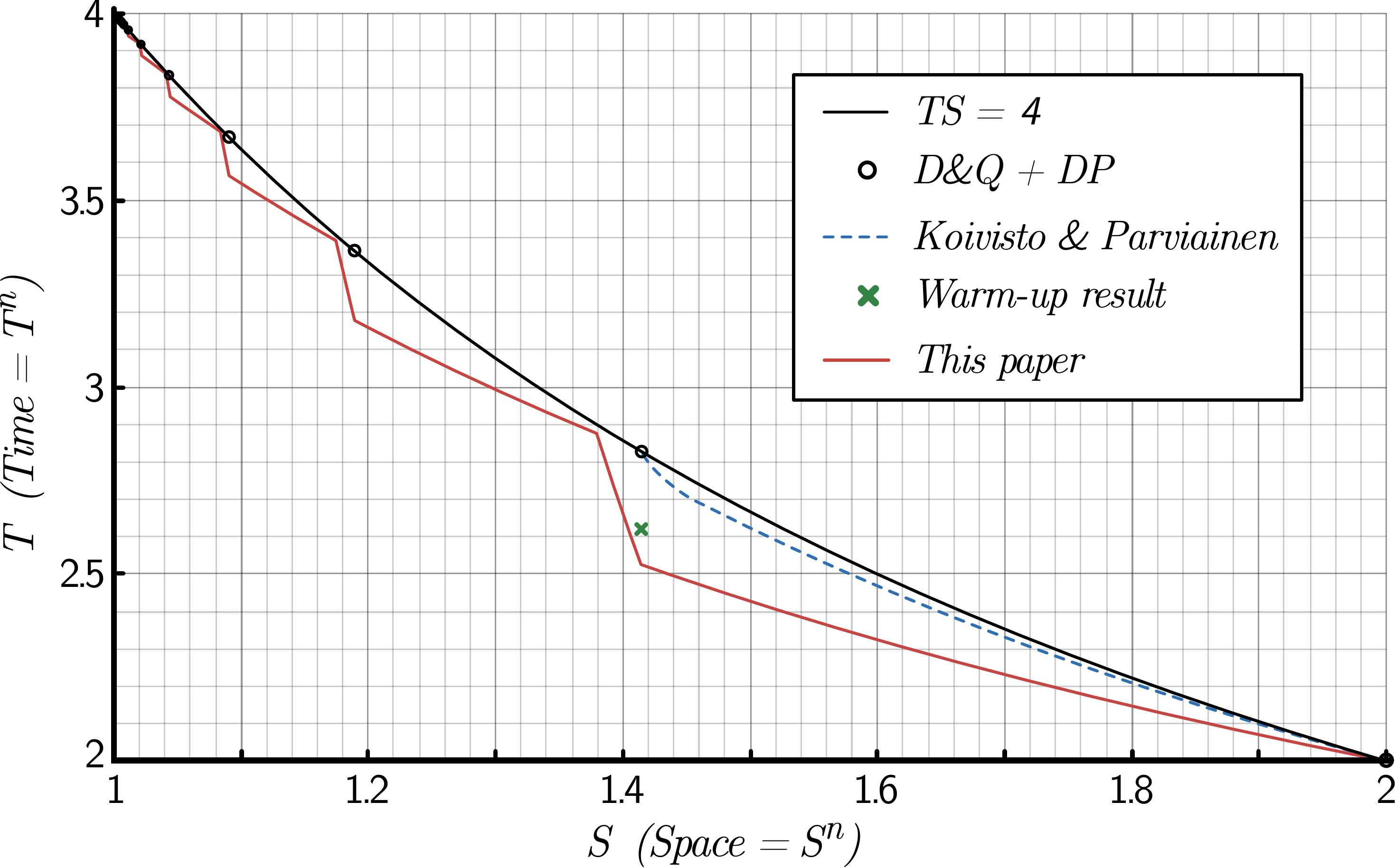}
    \caption{Space-time tradeoff for the TSP. Solid black line shows $ST=4$, with circles at the concrete feasible tradeoffs resulting from the classical divide-and-conquer and dynamic programming scheme. Dashed line is the improved tradeoff of Koivisto and Parviainen. The tradeoff achieved in our framework is shown in solid red; note that the curve does not touch $ST=4$, except at the endpoints $(1,4)$ and $(2,2)$. The green cross is the tradeoff point in our warm-up result in \S\,\ref{sec2}.}
    \label{fig1}
\end{figure}

The exact tradeoff curve we obtain is described in more detail in \S\,\ref{sec3}. It dominates all previous approaches and is illustrated in Figure~\ref{fig1}. Optimizing for a single point on the curve, we obtain:

\begin{theorem}
\label{thm2}
The pair $(S,T)$ with $S=\sqrt{2}$ and $T< 1.786 \cdot \sqrt{2}$, with $ST < 3.572$ is a feasible space-time tradeoff for TSP.
\end{theorem}

In \S\,\ref{sec2} we give a self-contained description of our tradeoff algorithm in a special case. This yields slightly weaker bounds than our stated results, but already improves all earlier ones, by a simple approach. In \S\,\ref{sec3} we describe our general algorithmic framework that leads to the proofs of Theorems~\ref{thm1} and \ref{thm2}. The framework reduces the algorithm design task to showing the existence of finite set systems with a high density of maximal chains. In \S\,\ref{sec4} we complete the proofs with the key ingredient: a concrete set system construction with the required property. In \S\,\ref{sec:lb} we show a lower bound on the best trade-off values that are attainable in our framework. In \S\,\ref{sec:comb} we discuss connections to set system (hypergraph) theory, and disprove the conjecture of Johnson, Leader, and Russell~\cite{JLRsetSystems} from 2013.  Finally, in \S\,\ref{sec5} we conclude with some open questions.  

\smallskip

Our general approach readily applies to the broader class of problems that Koivisto and Parviainen call \emph{permutation problems of bounded degree} (see also \cite{DBLP:journals/mst/BodlaenderFKKT12} where problems of this type appear as \emph{linear ordering problems}). Intuitively (we give a more formal definition in \S\,\ref{sec3}), for problems in this class, every permutation of $[n]$ has an associated cost that is, in some sense, decomposable into local costs, and we want to compute some aggregate value of these costs over all permutations (typically, the minimum). In the case of the TSP, the cost of a permutation is the length of the corresponding tour, which decomposes into the sum of distances of adjacent cities. Other problems in this class include the (directed) feedback arcset problem as well as the computation of the cutwidth, pathwidth and treewidth of a graph \cite{KoivistoParviainen2010, DBLP:journals/mst/BodlaenderFKKT12}. Note that Koivisto and Parviainen define this class of problems over semirings, and, while we improve on the earlier best tradeoffs for general semirings, e.g., for counting variants of various problems, our overall best bound requires the semiring to be \emph{additively idempotent} (this is the case in all the above examples, and whenever the additive operation of the semiring is $\min$, which obeys $\min(x,x) = x$ for all $x$). 

Our result in this direction can be stated as follows.

\begin{theorem}[Informal]\label{thm3}
The tradeoffs of Theorems~\ref{thm1} and \ref{thm2} are feasible for all permutation problems of constant degree over an additively idempotent semiring. Over a general semiring we obtain a tradeoff $(S,T)$ with $ST < 3.864$.
\end{theorem}

By making the tradeoff between prefix-sets and maximal chains in set systems explicit, our framework generalizes poset-based approaches and also gives a more intuitive view of the Koivisto-Parviainen result (we discuss this in Appendix~\ref{appa}). Overall, the algorithm design task is brought into a familiar setting of extremal combinatorics, where the task is to construct a (finite) combinatorial object with two parameters naturally in tension with each other.

To conclude our introduction, we mention that space-time tradeoffs have also been studied for other well-known problems. Influential results, based on rather different techniques, include the Fiat-Naor scheme for inverting functions~\cite{fiat2000rigorous} and the Schroeppel-Shamir scheme for Subset Sum~\cite{schroeppel1981t} with subsequent improvements~\cite{austrin2013space, bansal2018faster, NederlofW23}. 

\paragraph{Note on related work.} We recently learned that an improved space-time tradeoff for TSP has also been obtained independently and concurrently by Afrouz Jabal Ameli, Jesper Nederlof, and Shengzhe Wang. Their work appears on arXiv simultaneously with ours.

\section{Improved space-time tradeoff}
\label{sec2}

As a warm-up, we give a self-contained description of an algorithm (a special case of our general framework) with running time $\fO^*(2.6209^n)$ and space usage $\fO^*(1.4143^n)$, with $ST < 3.7063$.  This already improves the previous best result of $ST \approx 3.93$ and is achieved by a remarkably simple approach. Our algorithm is Monte Carlo randomized, yielding the optimal TSP tour with high probability. Later (\S\,\ref{sec3}) we show that the method can be derandomized without changing the bases of the exponentials and we also give a more general view of the algorithm, deriving improved bounds for the entire range of parameters. 

We first review the Bellman-Held-Karp dynamic programming algorithm. To solve TSP, it computes entries $\cT(S',c)$ for all $S' \subseteq S$ such that $c,c_1 \in S'$, indicating the minimum length of a tour (a Hamiltonian \emph{path}) that visits exactly the cities in $S'$, starting in $c_1$ and ending in $c$. (Notice that the choice of the starting city $c_1 \in S$ is arbitrary.) The overall optimum is then $$\OPT(S) = \min_{c \in S \setminus \{c_1\}}\{\cT(S,c) + d(c,c_1)\}.$$ Each entry of the table can be computed in $\fO(n)$ time with the recurrence $$\cT(S',c) = \min_{c' \in S'\setminus\{c_1,c\}}\{\cT(S'\setminus{c},c')+d(c',c)\},$$ 
with base cases $\cT(\{c_1,c\},c) = d(c_1,c)$, for all $c \in S \setminus \{c_1\}$. The number of table entries is upper bounded by $n \cdot 2^n$, yielding the total runtime $\fO(n^2 \cdot 2^n)$.

\medskip

Suppose we can guess a set $S' \subseteq S$ of cities with $|S'| = \lfloor n/2 \rfloor$ with the guarantee that the first $\lfloor \upalpha n \rfloor $ cities of the tour (assuming some canonical starting point $c_1$) are from $S'$, and that the last $\lfloor \upalpha n \rfloor$ cities of the tour (before returning to $c_1$) are from $S \setminus S'$. For this to be possible, we need $\upalpha \leq 1/2$; with foresight we pick $\upalpha \approx 0.445$.

Let $\cP$ denote the collection of prefix-sets that we must consider in the Bellman-Held-Karp DP. Then: $$|\cP| ~~\leq~~ 2^{|S'|} ~~+   \sum_{i=\lfloor \upalpha{n} \rfloor}^{|S'|}{|S'| \choose {i}} \cdot \sum_{i=0}^{|S \setminus S'| - \lfloor \upalpha n \rfloor}{|S \setminus S'| \choose {i}}  ~~+~~  2^{|S \setminus S'|}.$$

The first term corresponds to the first $\lfloor \upalpha n \rfloor$ cities in the tour. Since these are guaranteed to be from $S'$, we only need to consider subsets of $S'$ of which there are $2^{|S'|}$. Similarly, the last term corresponds to the last $\lfloor \upalpha n \rfloor$ cities in the tour. Since these are guaranteed to be from $S \setminus S'$, the set $S'$ must be part of the prefix, and we only need to additionally consider subsets of $S \setminus S'$  of which there are $2^{|S \setminus S'|}$.

The middle term corresponds to the portion of the tour that is between the first $\lfloor \upalpha n \rfloor$ and the last $\lfloor \upalpha n \rfloor$ cities. Here, $\lfloor \upalpha n \rfloor$ cities from $S'$ have been visited already, so they must be part of the prefix (of course we do not know which cities these are). The prefix must therefore contain a subset of at least this many cities from $S'$ (the first sum). Similarly, we may visit some cities from $S \setminus S'$, but only at most $|S \setminus S'| - \lfloor \upalpha n \rfloor$ of them (the second sum), such as to leave $\lfloor \upalpha n \rfloor$ for the last part. 

We apply the well-known identity ${n \choose k} = {n \choose n-k}$ and the standard estimate that, for $0 < \upbeta \leq 1/2$, yields $\sum_{i=0}^{\lfloor \upbeta n \rfloor}{{n \choose i} \leq 2^{n H(\upbeta)}}$, where $H(x) = -x \lg{x} - (1-x) \lg{(1-x)}$ is the binary entropy function\footnote{We denote $\lg{n} = \log_2{n}$.} (e.g., see~\cite[\S\,3.2]{FominK10}). We obtain $|\cP| \leq \fO^*(2^{n/2} + 2^{n H(2\upalpha)})$. Choosing $2\upalpha \approx 0.889972$, the root of $H(x) = \nicefrac{1}{2}$, we obtain $|\cP| \leq \fO^*(2^{n/2})$. The DP table has at most $n$ entries for each element of $\cP$, and each can be computed in $\fO(n)$ time as before, filling in entries in increasing order of the prefix-set-sizes. The bound on $|\cP|$ thus captures both the space and time requirements of the algorithm (assuming that the initial guess for $S'$ was correct). 

\smallskip

It remains to choose the set $S'$ with the above guarantees. We choose $S'$ uniformly at random among sets of cardinality $\lfloor n/2 \rfloor$. The probability $p$ that such a set fulfills the requirements is:
$$ p ~\geq~ \frac{\mybinom{n - 2 \lfloor \upalpha n \rfloor}{\lfloor n/2 \rfloor - \lfloor \upalpha n \rfloor}}{\mybinom{n}{\lfloor n/2 \rfloor}}. $$

Here the denominator counts the number of possible choices for $S'$ and the numerator captures the fact that for two disjoint subsets of size $\lfloor \upalpha n \rfloor$ the choice is fixed, and from the remaining cities exactly half should be picked in $S'$.

Using standard upper and lower bounds on the binomial coefficients, we obtain $1/p \in \fO^*(2^{n - n(1-2\upalpha)})$, which for our previous choice of $\upalpha$ yields $1/p \in \fO^*(1.8532^n)$. Picking $S'$ independently (say) $n/p$ times yields at least one with the required properties with high probability, and we can find it by taking the minimum over the obtained solutions. (Observe that all choices of $S'$ yield \emph{some} feasible solution.) 
The overall runtime is thus $\fO^*(T^n)$ with $T < 1.8532 \cdot \sqrt{2}$ and the space usage is $\fO^*(S^n)$ with $S = \sqrt{2}$, yielding $ST < 3.7063$. 

While this algorithm is randomized, it can be derandomized without changing the bases of the exponentials. We describe this in a more general setting in \S\,\ref{sec3}. We note that throughout the paper we focus, for simplicity, only on computing the \emph{value} of the optimal tour. Constructing the actual \emph{tour} can be achieved by standard modifications or black-box reductions.

\section{General framework}
\label{sec3}

In this section we develop our general approach for space-time tradeoffs for the TSP. Afterwards we extend the approach to a more general class of problems. The approach strongly builds on the Bellman-Held-Karp DP algorithm described in \S\,\ref{sec2}. 
We start with some definitions. 

\paragraph{Set systems.} A set system $\cF$ over $[n]$ is a collection of subsets of the set $[n]$, i.e., $\cF \subseteq 2^{[n]}$. We refer to $[n]$ as the ground set of $\cF$, and denote its cardinality by $n(\cF) = n$. Two set systems are said to be \emph{isomorphic} if one can be obtained from the other by a (bijective) relabeling of the elements of the ground set.

A \emph{maximal chain} of a set system $\cF$ over $[n]$ is a sequence $S_0, \dots, S_{n} \in \cF$, where $S_{i-1} \subsetneq S_{i}$ for $i \in [n]$.\footnote{In this paper we only call a chain maximal if it is maximal with respect to the full powerset $2^{[n]}$, i.e., of size $n+1$.} Notice that necessarily $|S_i| = i$, for all $i$, and in particular, $S_0 = \emptyset$ and $S_n = [n]$.

For a permutation $\pi$ of $[n]$, we let $\pi^{(0)} = \emptyset$, and for $i \in [n]$ define $\pi^{(i)} = \{\pi(1), \dots, \pi(i)\}$, the \emph{prefix-sets} of $\pi$. There is a natural bijection between permutations of $[n]$ and maximal chains of $2^{[n]}$, given by mapping a permutation to its sequence of prefix-sets. If a permutation $\pi$ is mapped to a maximal chain of $\cF$, i.e., if all prefix-sets of $\pi$ are in $\cF$, we say that $\pi$ is \emph{supported} by $\cF$.

A quantity that plays an important role in our study is the number of maximal chains of a set system $\cF$, denoted $C(\cF)$. Clearly, $C(\cF) \leq n!$, and it is natural to refer to $C(\cF)/n!$ as the \emph{chain-density} of $\cF$.
The space-time tradeoff we develop crucially depends on the existence of set systems $\cF$ of (relatively) small size and high chain-density. %
We soon make this requirement precise, but first we define normalized forms of the quantities in question.

\begin{definition}
Let $\cF$ be a set system over $[n]$. Then, the normalized size and the (inverse) normalized chain-density of $\cF$ are defined as follows. 
\begin{itemize}
    \item $S(\cF) = |\cF|^{1/n}$, and
        \item $P(\cF) = \left(\frac{n!}{C(\cF)}\right)^{1/n}$.
\end{itemize}
Note that in the latter definition we require $C(\cF) > 0$, and otherwise let $P(\cF) = +\infty$. %
\end{definition}

One can observe that the bounds $S(\cF) \leq 2$ and $P(\cF) \geq 1$ always hold.  

\paragraph{Algorithms.} The following key lemma is the bridge between set systems with a certain structure and feasible space-time tradeoffs. Perhaps surprisingly, the existence of a single finite set system that satisfies the condition is sufficient, and knowledge of which set system achieves this is not required. 

\begin{lemma}\label{lem:main}
Let $\cF$ be a set system with $S(\cF) \leq S$ and $S(\cF) \cdot P(\cF) \leq T$, for some $1 < S \leq 2$ and $S \leq T$. Then, there is an algorithm that solves TSP on inputs of size $n$, deterministically, in $T^{n+o(n)}$ time and $S^{n+o(n)}$ space.
\end{lemma}

More strongly, we can consider the extremal set systems of a given size over a given ground set that minimize the inverse normalized chain-density defined above. These set systems yield algorithms with a favorable space-time tradeoff.
Precisely, for any $1<S\leq 2$, denote 
\begin{align*}
P_S & = \inf\left\{P(\cF) \mid  S(\cF)\leq S\right\}, \text{~and}\\ 
P_S(n) & = \min\left\{P(\cF) \mid S(\cF)\leq S,~n(\cF)=n\right\}.
\end{align*}
Here, $P_S(n)=+\infty$ if the minimum is taken over an empty set. We now state our key theorem, proved in \S\,\ref{sec31}.

\begin{theorem}\label{thm:main}
    For any $1 < S \leq 2$, there is an algorithm that solves TSP on inputs of size $n$, deterministically, in $(S\cdot P_S)^{n+o(n)}$ time and $S^{n+o(n)}$ space.
\end{theorem}

Theorem~\ref{thm:main} clearly implies Lemma~\ref{lem:main}. Notice that if we only wish to find a single small $ST$ value, for some feasible tradeoff pair $(S,T)$, then the quantity of interest is $|\cF|^2 / C(\cF)$, which we should minimize for a given ground set size $n(\cF)$.

Given a feasible pair $(S,T)$, we can obtain other feasible pairs (with smaller $S$ and larger $T$) by combining the corresponding algorithm with the divide and conquer approach of Gurevich and Shelah, up to a certain level of recursion.
\begin{lemma}\label{lemma:partialDQ}
If $(S,T)$ is a feasible space-time tradeoff for TSP, then $(\sqrt{S},2\sqrt{T})$ is also feasible. 
\end{lemma}
\begin{proof}
Suppose that an algorithm $\mathcal{A}$ exists that solves TSP on inputs of size $n$ in time $T^{n + o(n)}$ and space $S^{n + o(n)}$. We guess the $\lfloor n/2 \rfloor$ cities making up the first half of the optimal tour and the cities beginning and ending both halves of the tour. This implies a factor of $\fO^*(2^n)$ in the running time. We then solve both halves recursively, using $\mathcal{A}$ as a black box, yielding the overall running time $2^{n+o(n)} \cdot T^{n/2}$ and space $S^{n/2 + o(n)}$, as required. A small omitted detail is that $\mathcal{A}$ finds a \emph{tour} (Hamiltonian \emph{cycle}), whereas in the recursive calls a Hamiltonian \emph{path} between two endpoints is sought for; this issue can be addressed identically to the Gurevich-Shelah algorithm in a black box manner.  
\end{proof}

Theorems~\ref{thm1} and \ref{thm2} in \S\,\ref{sec1} follow from Theorem~\ref{thm:main} and Lemma~\ref{lemma:partialDQ}, combined with a set system construction in \S\,\ref{sec4}.
Concretely, we show the existence of set systems that imply the bound $S^2\cdot P_S < 4$ for $1.38 < S < 2$ (see Figure~\ref{fig:LB}, solid red line), thus implying the claim of Theorem~\ref{thm1} over this range. In particular, for $S=\sqrt{2}$, we show that $P_S < 1.786$, implying Theorem~\ref{thm2}.

As we show in \S\,\ref{sec:lb}, $S^2\cdot P_S > 4$ when $S<1.15$ so the approach described up to now is not sufficient to prove Theorem~\ref{thm1} over the entire range $1<S<2$. However, by additionally applying Lemma~\ref{lemma:partialDQ} (up to $\fO(\log{n})$ times), we obtain a tradeoff curve that dominates $ST = 4$ over the entire range, as shown in Figure~\ref{fig1}. 

\subsection{Space-time tradeoffs from set systems}\label{sec31}

We state some additional definitions and lemmas that facilitate the proof of Theorem~\ref{thm:main}.

The next lemma follows directly from running the Bellman-Held-Karp algorithm as we did in the warm-up of \S\,\ref{sec2}, computing only the entries of the DP table corresponding to sets in $\cF$. 
\begin{lemma}
\label{thm:TSP_limited_prefixes}
    Given an instance of TSP on $n$ cities and a set system $\cF$ on $[n]$, one can compute the tour of least cost among those supported by $\F$ in $\fO^*(|\cF|)$ time and space.
\end{lemma}

The following operation is a useful way to construct larger set systems from smaller ones, while preserving the normalized quantities of interest.
\begin{definition}
    Given two set systems $\cF_1$ and $\cF_2$ over ground sets $[n_1]$ and $[n_2]$ respectively, the \emph{union product} of $\cF_1$ and $\cF_2$ (denoted by $\cF_1 \utimes \cF_2$) is the set system over $[n_1+n_2]$ defined as 
    \[
        \cF_1 \utimes \cF_2 = \{s_1 \cup (s_2+n_1) \mid s_1\in\cF_1, s_2\in \cF_2\},
    \]
    where  $(s_2+n_1)$ denotes the set $\{e+n_1 \mid e\in s_2\}$.
\end{definition}
Note that this operation is not commutative, although the set systems $\cF_1 \utimes \cF_2$ and $\cF_2 \utimes \cF_1$ are isomorphic.

We show that the union product affects the defined set system parameters in an intuitive way.

\begin{lemma}\label{lem37}
    Let $\cF_1,\cF_2,\ldots,\cF_k$ be set systems over ground sets $[n_1], [n_2], \ldots, [n_k]$ respectively, and let $\cF = \bigutimes_{i=1}^k \cF_i$.
    Then,     \begin{align*}
        n(\cF) &= \sum_{i=1}^k n_i,\\
        S(\cF) &= %
        \prod_{i=1}^k S(\cF_i)^{{n_i}/{n(\cF)}},\\
        P(\cF) &= 
        \prod_{i=1}^k P(\cF_i)^{n_i/n(\cF)}.
    \end{align*}
    In particular, if $S(\cF_i)\leq S$ and $P(\cF_i)\leq P$ for all $1\leq i\leq k$, then $S(\cF)\leq S$ and $P(\cF)\leq P$.
\end{lemma}
\begin{proof}
    The first equality is immediate from the definition of the union product. The second equality follows from observing that $|\cF| = \prod_{i=1}^k{|\cF_i|} = \prod_{i=1}^k{S(\cF_i)^{n_i}}$, since the ground sets of the set systems $\cF_i$ are shifted such as to make them disjoint, and there is a bijection between sets of $\cF$ and tuples $(s_1,\dots,s_k)$ where $s_i \in \cF_i$.

    For the third equality, we notice that there is a bijection between maximal chains of $\cF$ and tuples $(c_1, \dots, c_k)$ where $c_i$ is a maximal chain of $\cF_i$, once we fix the way these $k$ chains ``interleave'' to form the large chain. The number of possible ways to interleave them is $\binom{n(\cF)}{n_1,\dots,n_k} = \frac{n(\cF)!}{\prod_{i=1}^k{n_i!}}$. 
    It follows that $C(\cF) = \binom{n(\cF)}{n_1,\dots,n_k} \cdot \prod_{i=1}^{k}{C(\cF_i)}$, yielding the result by re-arranging and normalization.
\end{proof}

We also show that union products of set systems play along nicely with a natural notion of combining permutations, preserving the relation of \emph{support} between set systems and permutations. 

\begin{definition}
    Let $n>1$ and let $\pi$ be a permutation of $[n]$. For $n_1 > 0$ and $n_2 > 0$ such that $n_1+n_2 = n$, we call the pair $(\pi_1,\pi_2)$ the $(n_1,n_2)$-induced split of $\pi$, if
    \begin{itemize}
        \item $\pi_1$ is the permutation obtained from $\pi$ by ignoring all elements larger than $n_1$ (in other words, it is the permutation induced by $\pi$ on $[n_1]$),
        \item $\pi_2$ is is the permutation obtained from $\pi$ by ignoring all elements smaller or equal to $n_1$ and subtracting $n_1$ from the remaining elements.
    \end{itemize}

    Let $k>2$, and $n_1,n_2,\ldots,n_k$ be positive integers summing to $n$. 
    We define the $(n_1,n_2,\ldots,n_k)$-induced split $(\pi_1,\pi_2,\ldots,\pi_k)$ of $\pi$ recursively by letting $(\pi',\pi_k)$ be the $(n_1+n_2+\ldots+n_{k-1},n_k)$-induced split of $\pi$, and $(\pi_1,\pi_2,\ldots,\pi_{k-1})$ the $(n_1,n_2,\ldots,n_{k-1})$-induced split of $\pi'$.
\end{definition}

For example, if $n=7$ and $\pi= (1,4,3,6,2,5,7)$, then the $(2,2,3)$-induced split of $\pi$ is $(\pi_1,\pi_2,\pi_3)$, where $\pi_1 = (1,2)$, $\pi_2 = (2,1)$ and $\pi_3 = (2,1,3)$. 

\begin{lemma}\label{lemma:split}
    Let $\cF_1,\cF_2,\ldots,\cF_k$ be set systems over ground sets $[n_1], [n_2], \ldots, [n_k]$ respectively, let $\cF = \bigutimes_{i=1}^k \cF_i =  \cF_1\utimes\cF_2\cdots \utimes \cF_k$, and
    let $n=n(\cF)$ ($=n_1+\ldots +n_k$). 
    If $\pi$ is a permutation of $[n]$ and $(\pi_1,\pi_2,\ldots,\pi_k)$ is the $(n_1,n_2,\ldots,n_k)$-induced split of $\pi$, then $\pi$ is supported by $\cF$ if and only if $\pi_i$ is supported by $\cF_i$ for all $1\leq i\leq k$.
\end{lemma}
\begin{proof}
    We prove both directions for $k=2$ only, for $k>2$ we can iterate the argument through repeated splitting.
    
    For all $s \in \cF$, we have $s \cap [n_1] \in \cF_1$. Let $j_i$ be the position where $\pi_1(i)$ appears in $\pi$. Since $\pi$ is supported by $\cF$, we have $\pi^{(j_i)} \in \cF$, and therefore $\pi^{(j_i)} \cap [n_1] = \pi_1^{(i)} \in \cF_1$, for $0 \leq i \leq n_1$.
    
    Similarly, for all $s \in \cF$, we have $(s \cap ([n_2]+n_1))-n_1 \in \cF_2$. Let $j_i$ be the position where $\pi_2(i)+n_1$ appears in $\pi$. Since $\pi$ is supported by $\cF$, we have $\pi^{(j_i)} \in \cF$, and therefore $(\pi^{(j_i)} \cap ([n_2]+n_1))-n_1 = \pi_2^{(i)} \in \cF_2$, for $0 \leq i \leq n_2$. Therefore, $\pi_1$ and $\pi_2$ are supported by $\cF_1$ and $\cF_2$, respectively.

    In the reverse direction, since $\cF = \cF_1 \utimes \cF_2$, for all $s_1 \in \cF_1$ and $s_2 \in \cF_2$ we have $s_1 \cup (s_2 + n_1) \in \cF$, by the definition of the union product. 

    Consider the prefix-set $s = \pi^{(i)}$ for arbitrary $0 \leq i \leq n_1+n_2$. Let $s_1 = s \cap [n_1]$ and $s_2 = (s \cap ([n_2]+n_1))-n_1$. Since $s_1$ and $s_2$ are prefix-sets of $\pi_1$ and $\pi_2$, which in turn are supported by $\cF_1$ and $\cF_2$, we have $s_1 \in \cF_1$ and $s_2 \in \cF_2$. It follows that $s_1 \cup (s_2 + n_1) = s \in \cF$, therefore $\pi$ is supported by $\cF$. 
\end{proof}

Union products of set systems will be used explicitly as part of our final algorithm, but also allow us to prove the following useful result.
\begin{lemma}\label{lem:fekete}
    For any $1<S\leq 2$, the value $P_S(n)$ approaches $P_S$ as $n\to\infty$. In other terms, $P_S(n) = P_S+o(1)$.
\end{lemma}
\begin{proof}
    By Fekete's subadditive lemma~\cite{steele1997probability} applied to the sequence $a_n = \lg(P_S(n)^n)$, we have $\lim_{n \to \infty}{\frac{a_n}{n}} = \lim_{n \to \infty}{P_S(n)} = P_S$. It remains to show that $a_n$ is subadditive, i.e., that $a_{n+m} \leq a_n + a_m$.

    If $a_n = \infty$ or $a_m = \infty$, then this is immediate. If $a_n < \infty$ and $a_m < \infty$, then let $\cF_1$ (resp.\ $\cF_2$) be a set system on $[n]$ (resp.\ $[m]$) with $S(\cF_1) \leq S$ and $\lg(P(\cF_1)^n) = a_n$ (resp.\ $S(\cF_2) \leq S$ and $\lg(P(\cF_2)^m) = a_m$). Such set systems exist by definition of the sequence $(a_n)_{n\in \N}$.

    Let $\cF = \cF_1 \utimes \cF_2$. By Lemma~\ref{lem37}, $n(\cF)=n+m$, $S(\cF)\leq S$ and $\lg(P(\cF)^{n+m}) = \lg(P(\cF_1)^n) + \lg(P(\cF_2)^m) = a_n+a_m$. By definition of $P_S(n+m)$ and $a_{n+m}$, we have $a_{n+m} \leq a_n + a_m$.
\end{proof}

The following ``interpolation lemma'' is a simple consequence of Lemma~\ref{lem37} and Lemma~\ref{lem:fekete}:
\begin{lemma}\label{lemma:interpolate}
    For an arbitrary parameter $0\leq \mu \leq 1$, let  $1<S_1\leq 2$ and $1<S_2\leq 2$. Then $P_{S_1^\mu S_2^{1-\mu}} \leq P_{S_1}^\mu P_{S_2}^{1-\mu}$.
\end{lemma}

Finally, we move from focusing on the permutations supported by a single particular set system to all permutations of a certain size, supported via a family of set systems.

\begin{lemma}\label{thm:almost_optimal_chain_covering}
    For any $1<S\leq 2$ and any $n$ such that $S^n \geq n+1$, there is a family of $q(n)$ isomorphic set systems $\cF_1, \cF_2, \ldots, \cF_{q(n)}$ over $[n]$, such that:
    \begin{itemize}
        \item $q(n) \leq (P_S+o(1))^{n}$,
        \item $|\cF_j|\leq S^{n}$ for all $1\leq j \leq q(n)$, and
        \item for every permutation $\pi$ of $[n]$ there is some $1\leq j \leq q(n)$ such that $\cF_j$ supports $\pi$.
    \end{itemize}
\end{lemma}
\begin{proof}

Let $\cF$ be a set system that is extremal (minimizing) with respect to $P(\cF)$ among set systems of size $\lfloor S^n \rfloor$ over $[n]$ and denote $P = P_S(n) = P(\cF)$. Let $\cS$ be the set of permutations of $[n]$ supported by $\cF$. Recall from the definition of $P_S$, that $|\cS| = \frac{n!}{P^n}$ and notice that since $|\cF| \geq n+1$ and $\cF$ maximizes the number of supported permutations, it can be assumed that $|\cS|$ is nonzero.

Set $q(n) = P^n \cdot n^2$ and take $\cF_1, \dots, \cF_{q(n)}$ to be set systems isomorphic to $\cF$ obtained by relabeling the ground set $[n]$ according to $q(n)$ independently drawn uniform random permutations. Since $q(n) = (P+\fO(\frac{\lg{n}}{n}))^n$, and $P = P_S + o(1)$, by Lemma~\ref{lem:fekete}, the sequence of $q(n)$ set systems clearly satisfies the first two conditions. 

For an arbitrary permutation $\tau$ of $[n]$, the probability that $\cF_i$ supports $\tau$ (for arbitrary $i$) is $\frac{1}{P^n}$. This is because for each element of $\cS$ there is a unique permutation that maps it to $\tau$, so altogether $\cS$ permutations (out of the $n!$ total) lead to $\tau$ being supported. 

The probability of $\tau$ not being supported by any of $\cF_1, \dots, \cF_{q(n)}$ is at most $\left(1-\frac{1}{P^n}\right)^{q(n)} \leq e^{-\frac{q(n)}{P^n}} = e^{-n^2}$. By the union bound, the probability that there is \emph{some} permutation of $[n]$ not supported by any of $\cF_1, \dots, \cF_{q(n)}$ is thus at most $\frac{n!}{e^{n^2}}$.

Since this probability is strictly below $1$ for all positive $n$, by the probabilistic method, there is a sequence which satisfies the first two conditions, and where for every permutation $\pi$, at least one set system in the sequence supports $\pi$. 
\end{proof}

We are now ready to prove the main theorem.
\begin{proof}[Proof of Theorem \ref{thm:main}]
Let $m = \floor{\frac{\lg\lg\lg n}{2}}$, let $k=\floor{\frac{n}{m}}$ and for $1\leq i \leq k$, let $m \leq n_i \leq 2m$, such that $\sum_{i=1}^kn_i=n$.

 For $1\leq i \leq k$, consider a family of set systems $\F^{i}_1, \F^{i}_2, \ldots, \F^{i}_{q_i}$ on $[n_i]$, such that 
\begin{itemize}
    \item $|\F^{i}_j|\leq S^{n_i}$ for all $1\leq j \leq q_i$, 
    \item for every permutation $\pi_i$ of $[n_i]$ there is some $1\leq j \leq q_i$ such that $\F^{i}_j$ supports $\pi_i$, and
    \item $q_i$ is as small as possible.
\end{itemize}

By Lemma~\ref{thm:almost_optimal_chain_covering}, for $m$ large enough (such that $S^m \geq m+1$), there is such a family of set systems with $q_i = (P_S+o(1))^{n_i}$. 

Let $\pi$ be a permutation of $[n]$, and $(\pi_1, \pi_2, \ldots, \pi_k)$ its $(n_1,n_2,\dots,n_k)$-induced split. For every $1\leq i \leq k$, there is some $1\leq j_i \leq q_i$ such that $\cF_{j_i}^i$ supports $\pi_i$. By Lemma~\ref{lemma:split}, $\pi$ is supported 
by $\bigutimes_{i=1}^k \cF^i_{j_i}$.

Assume for now that for all possible choices of $\mathbf{j} = (j_1, j_2, \ldots, j_k)$, where $1\leq j_i\leq q_i$,  for $i=1,\dots,k$, we can compute $\cF_{\mathbf{j}} = \bigutimes_{i=1}^k \cF^i_{j_i}$ in time and space $\fO^*(S^n)$. 

For any choice of $\mathbf{j}$,  we have $|\F_{\mathbf{j}}| \leq S^{\sum n_i}=S^n$ and by Lemma~\ref{thm:TSP_limited_prefixes} we can find the optimal tour, assuming that the corresponding ordering of the cities $\pi$ is such that all its prefix-sets are in $\F_{\mathbf{j}}$, in $\fO^*(S^n)$ time and space.

There are $\prod_{i=1}^k q_i \leq (P_S+o(1))^{n} \leq P_S^{n+o(n)}$ possible choices for $\mathbf{j}$, and for every permutation $\pi$ of $[n]$ there is at least one choice of $\mathbf{j}$ such that $\pi$ is supported by $\F_{\mathbf{j}}$. Thus, by repeating the above for every possible choice of $\mathbf{j}$, we can find the optimal tour in $\fO^*(S^n)$ space and $\fO((S\cdot P_S)^{n+o(n)})$ time.

It remains to show how to compute $\cF_{\mathbf{j}}$ in $\fO^*(S^n)$ time and space.

For all $1\leq i \leq k$, we precompute $q_i$ and $\F^{i}_1, \F^{i}_2, \ldots ,\F^{i}_{q_i}$ by brute-force, considering all families of distinct set systems on $[n_i]$ with set systems of sizes at most $S^{n_i}$ and test each of them against %
all permutations of $[n_i]$. For a given $1\leq i \leq k$, this amounts to at most $2^{2^{2m}} \in \fO(\lg{n})$ set systems and $(2m)! \in \fO(\lg{n})$ permutations. %
Finding the smallest family of set systems that together support all permutations of $[n_i]$ amounts to solving a set cover problem. Since the set cover instance is of size $\fO(\lg{n})$, this can be carried out in time and space polynomial in $n$. 
The total time and space to achieve this for all $1\leq i \leq k$ is still polynomial in $n$. 

(We note that more efficient computation of the set systems $\F_{1}^{i}, \dots, \F_{q_i}^{i}$, and relaxing $m$ to $\approx \lg\lg{n}$ is possible, if we additionally use that the set systems are isomorphic, as guaranteed by Lemma~\ref{thm:almost_optimal_chain_covering}, or if we settle for solving the set cover problem approximately. We forgo such optimizations as they are not consequential to our main claim.)

For a given $\mathbf{j} = (j_1, j_2, \ldots, j_m)$, computing $\F_{\mathbf{j}}$ can then be done in $\fO^*(S^n)$ time and space by computing $\bigcup_{i=1}^k C_i$ for all $C_1 \in \F^{1}_{j_1}, C_2 \in \F^{2}_{j_2}, \ldots, C_k \in \F^{k}_{j_k}$.

Clearly, all steps of the computation can be performed deterministically.
\end{proof}

\subsection{Generalizing to permutation problems}

We briefly discuss how our approach applies more generally to so-called permutation problems (a.k.a.\ linear ordering problems).
We start by recalling the definition, largely following~\cite{KoivistoParviainen2010}.
\begin{definition}
    Let $n>0$, let $d\geq 0$, and $R$ be a semiring with addition $\oplus$ and multiplication $\otimes$.

    Let $f$ be a cost function that maps permutations of $[n]$ to values in $R$, decomposable into local cost functions $f_i$ as follows:
    \[f(\pi) = \bigotimes_{j=1}^n f_j(\{\pi_1,\pi_2,\ldots,\pi_j\}, \pi_{j-d+1},\ldots, \pi_{j-1}, \pi_j).\]
    If $d>j$, the sequence $\pi_{j-d+1},\ldots, \pi_{j-1}, \pi_j$ is to be read as $\pi_1, \pi_2,\ldots, \pi_j$, and if $d=0$, it is empty.

    We call \emph{permutation problem of degree $d$} the task of computing $\bigoplus_\pi f(\pi)$ for $\pi$ ranging over all permutations of $[n]$. For such a problem we assume that the local costs $f_i$ and the operations $\oplus$ and $\otimes$ are computable in polynomial time.

\end{definition}

For example, TSP reduces to finding a minimum weight Hamiltonian path in a weighted graph, which is a permutation problem of degree two over the ($\min$,$+$) semiring with $f_1(A,x)=0$ and $f_j(A,x,y)$ being equal to the weight of the edge $xy$ for $j > 1$. Other examples include the (directed) feedback arcset problem as well as the computation of the cutwidth, pathwidth and treewidth of a graph (see \cite{KoivistoParviainen2010, DBLP:journals/mst/BodlaenderFKKT12}).

The Bellman-Held-Karp algorithm for TSP, as well as the variant of Lemma~\ref{thm:TSP_limited_prefixes} restricting the considered permutations by their prefix-sets, both easily extend to permutation problems of constant degree. 

The approach we develop in this paper, based on grouping permutations by their prefix-sets, also applies in this setting but with a caveat: because a permutation might appear in multiple groups, it can contribute multiple times to the output of the algorithm, whereas it only contributes once to the correct output $\bigoplus_\pi f(\pi)$. In the special case where the operation $\oplus$ is idempotent (i.e., $x\oplus x = x$ for all $x\in R$), or in other words when the semiring $R$ is additively idempotent, then contributing multiple times does not change the output. Our general approach thus directly applies to this case, with time- and space bounds unchanged, up to a polynomial factor. Note that all the examples of permutation problems mentioned above fall into this category.

The only obstacle to applying our approach to the non-idempotent case is Lemma~\ref{thm:almost_optimal_chain_covering}, where a permutation $\pi$ might be supported by $\cF_j$ for multiple different $j$. In the rest of this section we show how to obtain a version of this lemma where every permutation is uniquely supported, at the cost of restricting the family of set systems that are admissible.

We start with some definitions.

\begin{definition}
    Two set systems $\cF_1$ and $\cF_2$ on the same ground set are \emph{regularly intersecting} if there is a subset $\mathcal{G} \subseteq \cF_1 \cap \cF_2$ such that all permutations supported by both $\cF_1$ and $\cF_2$ have at least one prefix in $\mathcal{G}$, and all permutations supported by $\cF_1$ but not by $\cF_2$ have no prefix in $\mathcal{G}$.

    We say that $\cF$ is \emph{regularly self-intersecting} if for all $\cF'$ isomorphic to $\cF$, the set systems $\cF$ and $\cF'$ are regularly intersecting.
\end{definition}

    For any $1<S\leq 2$,  we denote the following extremal quantities, analogously to our earlier definitions for general set systems:
    \begin{align*}
    P'_S & = \inf\left\{P(\cF) \mid  S(\cF)\leq S\right\}, \text{~and}\\ 
    P'_S(n) & = \min\left\{P(\cF) \mid S(\cF)\leq S,~n(\cF)=n\right\},
    \end{align*}
    where $\cF$ ranges over regularly self-intersecting set systems and $P'_S(n)=+\infty$ if the minimum is taken over an empty set.

\begin{lemma}
    If $\cF_1$ and $\cF_2$ are regularly self-intersecting set systems then $\cF = \cF_1 \utimes \cF_2$ is regularly self-intersecting.
\end{lemma}
\begin{proof}
    Let $\cF'$ be a set system isomorphic to $\cF$. It can be written as $\cF' = \cF'_1 \utimes \cF'_2$ where $\cF'_1$ and $\cF'_2$ are isomorphic to $\cF_1$ and $\cF_2$ respectively. 

    Because $\cF_1$ is regularly self-intersecting, there is a subset $\mathcal{G}_1 \subseteq \cF_1 \cap \cF'_1$ such that all permutations supported by both $\cF_1$ and $\cF'_1$ have a prefix in $\mathcal{G}_1$ and all permutations supported by $\cF_1$ but not $\cF'_1$ have no prefix in $\mathcal{G}_1$.

    Similarly, there is such a subset $\mathcal{G}_2 \subseteq \cF_2 \cap \cF'_2$ for $\cF_2$ and $\cF'_2$.

    All permutations supported by both $\cF$ and $\cF'$ have a prefix in $\mathcal{G}_1 \utimes \mathcal{G}_2$ and all permutations supported by $\cF$ but not $\cF'$ have no prefix in $\mathcal{G}_1 \utimes \mathcal{G}_2$. Thus, since $\mathcal{G}_1 \utimes \mathcal{G}_2 \subseteq \cF \cap \cF'$, the set systems $\cF$ and $\cF'$ are regularly intersecting, and $\cF$ is regularly self-intersecting.
\end{proof}

Using the above definitions and lemma, we mimic the steps taken to prove Lemma~\ref{thm:almost_optimal_chain_covering},  and obtain the following analogous result for regularly self-intersecting set systems.
\begin{lemma}\label{thm:almost_optimal_regular_chain_covering}
     For any $1<S\leq 2$ and any $n$ such that $S^n \geq n+1$, there is a family of $q(n)$ isomorphic regularly self-intersecting set systems $\cF_1, \cF_2, \ldots, \cF_{q(n)}$ over $[n]$, such that
    \begin{itemize}
        \item $q(n) = (P'_S+o(1))^{n}$,
        \item $|\cF_j|\leq S^{n}$ for all $1\leq j \leq q(n)$, and 
        \item for every permutation $\pi$ of $[n]$ there is some $1\leq j \leq q(n)$ such that $\F_j$ supports $\pi$.
    \end{itemize}
\end{lemma}

We next modify the set systems to ensure that every permutation is supported by only one of them.
\begin{lemma}
     For any $1<S\leq 2$ and any $n$ such that $S^n \geq n+1$, there is a family of $q(n)$ set systems $\cF_1, \cF_2, \ldots, \cF_{q(n)}$ over $[n]$, such that
    \begin{itemize}
        \item $q(n) = (P'_S+o(1))^{n}$,
        \item $|\cF_j|\leq S^{n}$ for all $1\leq j \leq q(n)$, and 
        \item for every permutation $\pi$ of $[n]$, there is \emph{a unique} $1\leq j \leq q(n)$ such that $\F_j$ supports $\pi$.
    \end{itemize}
\end{lemma}
\begin{proof}
    Start with a family $\cF'_1, \cF'_2, \ldots ,\cF'_{q(n)}$ given by Lemma~\ref{thm:almost_optimal_regular_chain_covering}.

    Next, let $\cF_1 = \cF'_1$ and for $i$ ranging from $2$ to $q(n)$, define $\cF_i$ as follows:
    \begin{itemize}
        \item For $1\leq k < i$, let $\mathcal{G}_i^k\subseteq \cF'_i\cap \cF'_k$ be such that all permutations supported by both $\cF'_i$ and $\cF'_k$ have a prefix in $\mathcal{G}_i^k$, and no permutation supported by $\cF_i$ but not by $\cF_k$ has a prefix in $\mathcal{G}_i^k$ ($\mathcal{G}_i^k$ exists by definition of regularly self-intersecting set systems).
        \item Let $\mathcal{G}_i = \bigcup_{k=1}^i \mathcal{G}_i^k$.
        \item Let $\cF_i = \cF'_i \setminus \mathcal{G}_i$.
    \end{itemize}

    For all $i$, no permutation supported by $\cF_i$ is supported by any $\cF_k$ with $k<i$, because such a permutation would have a prefix in $\mathcal{G}_i^k$ and thus not be supported by $\cF_i$. In other words, no permutation is supported by more than one set system.

    On the other hand, let $\pi$ be a permutation of $[n]$, and suppose it is not supported by any of the set systems. Let $i>1$ be the smallest index such that $\pi$ is supported by $\cF'_i$ but not $\cF_i$. Then $\pi$ must have a prefix in $\mathcal{G}_i$, which has to be in $\mathcal{G}^k_i$ for some $k<i$. By definition of $\mathcal{G}^k_i$, and because $\pi$ is supported by $\cF'_i$, $\pi$ has to be supported by $\cF'_k$. By assumption, $\pi$ is not supported by $\cF_k$, which contradicts the minimality of $i$. We conclude by contradiction that $\pi$ is supported by at least one of the set systems.

    In short, every permutation of $[n]$ is supported by exactly one set system in $\cF_1, \cF_2, \ldots, \cF_{q(n)}$.
\end{proof}

A similar proof to that of Theorem~\ref{thm:main} gives the following.

\begin{theorem}\label{thm:main-permutation-problem}
    For any permutation problem of constant degree and any $1 < S \leq 2$, there is an algorithm that solves the problem on inputs of size $n$, deterministically, in $(S\cdot P'_S)^{n+o(n)}$ time and $S^{n+o(n)}$ space.
\end{theorem}

In the next section, we will show that for $S = 1.7913$, $P'_S < 1.20398$, thus implying the following.

\begin{theorem}\label{thm319}
    The pair $(S,T)$ with $S = 1.7916$, $T = S\cdot P'_S < 2.1567$ and $S\cdot T < 3.864$ is a feasible space-time tradeoff for permutation problems of constant degree.
\end{theorem}

Theorem~\ref{thm319} immediately implies Theorem~\ref{thm3} stated in the introduction. As discussed before, Theorem~\ref{thm319} applies to a variety of permutation problems of constant degree over arbitrary semirings.
Perhaps the most natural problems in this class, to which our stronger Theorems~\ref{thm1} and \ref{thm2} do not apply, are \emph{counting} problems. 

\medskip

Let us give a single representative application, the problem of \emph{counting linear extensions of posets} (\#\textsc{LE}). In this problem, given an input poset $([n],\prec)$, we seek the number of total orders (permutations of $[n]$) that contain (extend) $\prec$. 

The problem has been extensively studied, e.g., see~\cite{brightwell1991counting, stanley1986two} and references thereof. 
Algorithms analogous to the TSP dynamic program (resulting in $\fO^*(2^n)$ time and space) are applicable to \#\textsc{LE} with minimal changes (e.g., see~\cite{knuth1974structured}), with prefix-sets corresponding to downsets of the input poset. Both the Gurevich-Shelah divide-and-conquer, and the tradeoff scheme of Koivisto and Parviainen easily apply, yielding the same exponential bounds; to our knowledge -- apart from special cases (e.g., \cite{mohring1989computationally, kangas2020faster, felsner2015linear, kozma_poset}) -- no better tradeoff $ST$ is known for this problem.   

\#\textsc{LE} is a permutation problem of degree two over the $(+,\cdot)$ semiring, where we aim to compute $\sum_\pi f(\pi)$ over all permutations of $[n]$. Here $f(\pi)$ should be $1$ if $\pi$ is a valid linear extension of the input and $0$ otherwise; accordingly, $f_1(A,x) = 1$, and $f_j(A,x,y) = 0$ if $y \prec x$ and $1$ otherwise. As a consequence, Theorem~\ref{thm319} implies a new space-time tradeoff for this problem with the improved value $ST < 3.864$.

\section{Extremal set systems}
\label{sec4}

In this section we focus on the combinatorial problem of finding small set systems supporting many permutations, thus proving bounds for our framework. In Appendix~\ref{appa} we give some simpler examples of set systems, including those that imply the previous tradeoff results and our warm-up example from \S\,\ref{sec2}.

Our best construction will result in the following bounds.
\begin{theorem}\label{thm:construction}
    Let $\epsilon > 0$, $\nicefrac{1}{4}\leq \beta \leq \gamma \leq \nicefrac{1}{2}$ and $\beta \leq \alpha \leq \nicefrac{1}{2}$. 
    Then there is a set system $\cF$ with \begin{align*}
    \lg S(\cF) &\leq
    \max\left\{\alpha,{\frac{1}{2}\left(H(2\beta)+H(1-2\gamma)\right)}\right\} +\epsilon,
    \\
    \lg P(\cF) &\leq
    1+H(2\alpha) - \left(\frac{1}{2}-\beta\right)\left(H\left(\frac{\gamma-\beta}{\frac{1}{2}-\beta}\right) + H\left(\frac{\frac{1}{2}-\gamma}{\frac{1}{2}-\beta}\right) + 2H\left(\frac{\alpha-\beta}{\frac{1}{2}-\beta}\right) \right)
    +\epsilon.
\end{align*}
\end{theorem}

In particular, letting $\alpha=\nicefrac{1}{2}-o(1)$, $\beta=0.4112$ and $\gamma \approx 0.4703+o(1)$ be the root of $H(2\beta)+H(1-2\gamma) = 2\alpha$, yields the following.
\begin{corollary}
    For $S=\sqrt{2}$, $P_S < 1.785975$.
\end{corollary}

Letting $\alpha=0.46-o(1)$, $\beta=0.406$ and $\gamma \approx 0.4821+o(1)$ be the root of $H(2\beta)+H(1-2\gamma) = 2\alpha$ yields the following.
\begin{corollary}
    For $S=2^{0.46}$, $P_S < 2.121604$.
\end{corollary}

We can use Lemma \ref{lemma:interpolate} to interpolate between these two points (as well as the trivial point $P_S=1$ for $S=2$) to get the following.
\begin{corollary}
    For $0.46 \leq x \leq \nicefrac{1}{2}$ and $S=2^x$, we have $P_S < 1.785975\cdot74.0839^{\left(\nicefrac{1}{2}-x\right)}$.
    
    For $\nicefrac{1}{2}\leq x \leq 1$ and $S=2^x$, we have $P_S < 3.18971^{1-x}$.
\end{corollary}

In order to apply our approach to permutation problems over semirings which are not additively idempotent, we need the set systems we consider to have some additional properties.
As discussed in the previous section, a sufficient property is for the set system to be regularly self-intersecting. The following result shows that even with this restriction, we can improve on the previous best tradeoff with our framework.

\begin{theorem} \label{thm:construction-regular}
    Let $\epsilon > 0$.
    For any $0 < \alpha \leq 1$,  and $\frac{\alpha}{2} \leq \beta\leq \alpha$ there is a regularly self-intersecting set system $\cF$ with
    \begin{align*}
        \lg S(\cF) &\leq \max\left\{\alpha, (1-\alpha) + \alpha H\left(\frac{\beta}{\alpha}\right)\right\} +\epsilon, \\
        \lg  P(\cF) &\leq H(\alpha) - (1-\beta)H\left(\frac{\alpha-b}{1-\beta}\right) +\epsilon.
    \end{align*}
\end{theorem}
In particular, letting $\alpha = 0.8412$ and $\beta = 0.75 \cdot 0.8412$, yields the following.
\begin{corollary}
    For $S=1.7916$, $P'_S < 1.20375$.
\end{corollary}

To count the number of permutations supported by a particular set system $\F$, it will be more convenient in the cases which lead to this bound to reason about how many set systems isomorphic to $\F$ there are, and how many of those contain a particular permutation (similarly to how the counting was done in the warm-up of \S\,\ref{sec2}). The following lemma makes this translation explicit.

\begin{lemma}\label{lemma:supported_fraction}
    Let $n\geq 1$ and let $\cF$ be a set system on $[n]$. Let $N$ be the number of set systems on $[n]$ isomorphic to $\F$, and $M$ be the number of those supporting the identity permutation on $[n]$. Then $\F$ supports exactly $\frac{M}{N}n!$ permutations of $[n]$.
\end{lemma}
\begin{proof}
    This follows easily from a double counting argument. Every permutation on $[n]$ (and in particular the identity) is supported by the same number $M$ of set systems isomorphic to $\cF$, and every set system isomorphic to $\cF$ supports the same number, call it $Q$, of permutations on $[n]$. 
    We have $M\cdot n! = Q\cdot N$, as both hand sides count the number of pairs of set systems isomorphic to $\cF$ and permutations supported by said set systems. Rearranging, we get $Q=\frac{M}{N}n!$.
\end{proof}

We are now ready to prove Theorem~\ref{thm:construction} and Theorem~\ref{thm:construction-regular}.
\begin{proof}[Proof of Theorem~\ref{thm:construction}]
    Let $n\geq 1$, $\nicefrac{1}{4}\leq \beta \leq \gamma \leq \nicefrac{1}{2}$, and $\beta \leq \alpha \leq \nicefrac{1}{2}$. To simplify notation, we assume $\frac{n}{2}$, $\alpha n$, $\beta n$, and $\gamma n$ are integers. 

    Let $L_1,R_1$ be a partition of $[n]$ into two subsets of size $\frac{n}{2}$, and $L_2\subseteq L_1$, $R_2\subseteq R_1$ be two subsets of size $\alpha n$.
    
    We define $\cF$ as the minimal set system over $[n]$ supporting all permutations $\pi$ of $[n]$ with the following properties:
    \begin{itemize}
        \item the first $\beta n$ entries of $\pi$ are in $L_2$,
        \item the last $\beta n$ entries of $\pi$ are in $R_2$,
        \item among the first $\frac{n}{2}$ entries of $\pi$, at least $\gamma n$ are in $L_1$,
        \item among the last $\frac{n}{2}$ entries of $\pi$, at least $\gamma n$ are in $R_1$.
    \end{itemize}

    Note in particular that the conditions imply that among the first (resp.\ last) $k$ entries of such a permutation, at least $k-\frac{n}{2}+\gamma$ are in $L_1$ (resp.\ $R_1$).

    Let us estimate the number of sets in $\cF$. By definition, each of these sets is a prefix-set of a permutation with the above properties. 
    
    The number of sets in $\cF$ of size at most $\beta n$ is at most $2^{\alpha n}$, as these are subsets of $L_2$. The same holds for the number of sets in $\cF$ of size at least $n-\beta n$, as these are complements of subsets of $R_2$.

    By the constraints on the supported permutations, a set $s$ in $\cF$ of size $k$ between $\beta n$ and $\frac{n}{2}$ must contain at least $k_1\geq \max\{\beta n, k-\frac{n}{2}+\gamma n\} \geq \beta n$ elements from $L_1$. Similarly, the complement of $s$ must contain at least $\max\{\beta n, n-k-\frac{n}{2}+\gamma n\}$ elements from $R_1$. In other words, $s$ contains at most $k_2\leq \frac{n}{2} - \max\{\beta n, n-k-\frac{n}{2}+\gamma n\} \leq \frac{n}{2}-\gamma n$ elements from $R_1$.

    One can then bound the number of such sets by
    \begin{align*}
        \sum_{k_1=\beta n}^{n/2}\sum_{k_2=0}^{n/2-\gamma n} \binom{n/2}{k_1} \binom{n/2}{k_2} = \fO^*\left( \binom{n/2}{\beta n} \binom{n/2}{n/2-\gamma n}\right),
    \end{align*}
    where we have bounded the sum by its maximum term, using the fact that $\beta n \geq \frac{n}{4}$ and $\gamma n \geq \frac{n}{4}$.

    The same holds for sets in $\cF$ of size between $\frac{n}{2}$ and $n-\beta n$, by considering their complements.

    Thus, we have $|\cF| = \fO^*\left(2^{\alpha n} + \binom{n/2}{\beta n} \binom{n/2}{n/2-\gamma n}\right)$. 

    Let us now estimate the fraction of set systems isomorphic to $\cF$ which support the identity permutation. This is equivalent to estimating the fraction of ways to choose $L_1,R_1,L_2$ and $R_2$ which lead to a set system supporting the identity permutation.

    In total, there are $\binom{n}{n/2}\binom{n/2}{\alpha n}^2$ ways to choose $L_1,R_1,L_2$ and $R_2$. The identity permutation is supported when this choice is such that:
    \begin{itemize}
        \item $L_2$ contains ${1,2,\ldots, \beta n}$,
        \item $R_2$ contains ${n,n-1,\ldots, n-\beta n+1}$,
        \item $L_1$ contains at least $\gamma n$ elements from ${1,2,\ldots, \frac{n}{2}}$,
        \item $R_1$ contains at least $\gamma n$ elements from ${n,n-1,\ldots, \frac{n}{2}+1}$.
    \end{itemize}

We can generate such choices by first letting $L_1$ be a set containing ${1,2,\ldots, \beta n}$ together with $\gamma n-\beta n$ elements from  ${\beta n+1, \beta n+2, \ldots, \frac{n}{2}}$ and $\frac{n}{2}-\gamma n$ elements from  ${\frac{n}{2}+1, \frac{n}{2}+2, \ldots, n-\beta n}$. There are $\binom{n/2-\beta n}{\gamma n - \beta n}\binom{n/2-\beta n}{n/2 - \gamma n}$ ways to do so (and this choice also fixes $R_1$). Then choose for $L_2$ any subset of $L_1$ of size $\alpha n$ containing ${1,2,\ldots, \beta n}$ (there are $\binom{n/2-\beta n}{\alpha n - \beta n}$ choices), and choose for $R_2$ any subset of $R_1$ of size $\alpha n$ containing ${n,n-1,\ldots, n-\beta n+1}$ (again there are $\binom{n/2-\beta n}{\alpha n - \beta n}$ choices).

There are thus at least 
$\binom{n/2-\beta n}{\gamma n - \beta n}\binom{n/2-\beta n}{n/2 - \gamma n} \binom{n/2-\beta n}{\alpha n - \beta n}^2$ ways to choose $L_1,R_1,L_2$ and $R_2$ such that the identity permutation is supported by the resulting set system.

It follows that the fraction of set systems isomorphic to $\cF$ which support the identity permutation is at least 
\[f = \frac{\binom{n/2-\beta n}{\gamma n - \beta n}\binom{n/2-\beta n}{n/2 - \gamma n} \binom{n/2-\beta n}{\alpha n - \beta n}^2}{\binom{n}{n/2}\binom{n/2}{\alpha n}^2}.\]

By Lemma \ref{lemma:supported_fraction}, the number of permutations of $[n]$ supported by $\cF$ is  least $f\cdot n!$, and thus $P(\cF) \leq (1/f)^{1/n}$.

By estimating the binomial coefficients through the binary entropy function in the standard way, we get
\begin{align*}
    \lg S(\cF) &\leq
    \max\left\{\alpha,{\frac{1}{2}\left(H(2\beta)+H(1-2\gamma)\right)}\right\} +o(1),
    \\
    \lg P(\cF) &\leq
    1+H(2\alpha) - \left(\frac{1}{2}-\beta\right)\left(H\left(\frac{\gamma-\beta}{\frac{1}{2}-\beta}\right) + H\left(\frac{\frac{1}{2}-\gamma}{\frac{1}{2}-\beta}\right) + 2H\left(\frac{\alpha-\beta}{\frac{1}{2}-\beta}\right) \right)
    +o(1).
\end{align*}
For large enough $n$ we thus have the sought result.
\end{proof}

\begin{proof}[Proof of Theorem~\ref{thm:construction-regular}]
    Let $n\geq 1$, $0\leq \alpha \leq 1$, and $\frac{\alpha}{2} \leq \beta \leq \alpha$. To simplify notation, we assume $\alpha n$ and $\beta n$ integers. 
    Let $L$ be a subset of $[n]$ of size $\alpha n$.
    
    We define $\cF$ as the minimal set system over $[n]$ supporting all permutations $\pi$ of $[n]$ with the property that the first $\beta n$ entries of $\pi$ are in $L$. 
    Note that $\cF$ consists of all subsets of $L$ of size $\beta n$ together with all subsets and supersets of such sets.
    
    Estimating $S(\cF)$ and $P(\cF)$ can be done in the same way as in the previous proof, yielding the claimed counts, but we still need to argue that $\cF$ is regularly self-intersecting. 
    
    Let $\cF'$ be a set system isomorphic to $\cF$. There is $L'\subseteq [n]$ such that $\cF'$ consists of all subsets of $L'$ of size $\beta n$ together with all subsets and supersets of such sets. Let $\mathcal{G}$ consist of all subsets of $L\cap L'$ of size $\beta n$. Then, all permutations supported by both $\cF$ and $\cF'$ have a prefix in $\mathcal{G}$, and none of the permutations supported by $\cF$ but not $\cF'$ do.
    Thus $\cF$ and $\cF'$ are regularly intersecting and $\cF$ is regularly self-intersecting. 
\end{proof}

\subsection{Lower bound}
\label{sec:lb}

Recall that for a set system $\cF$, the quantities $S(\cF)$ and $P(\cF) \cdot S(\cF)$ correspond to the $S$, $T$ values in the feasible space-time tradeoff resulting from this set system. 

\begin{figure}
    \centering
\includegraphics[width=0.85\linewidth]{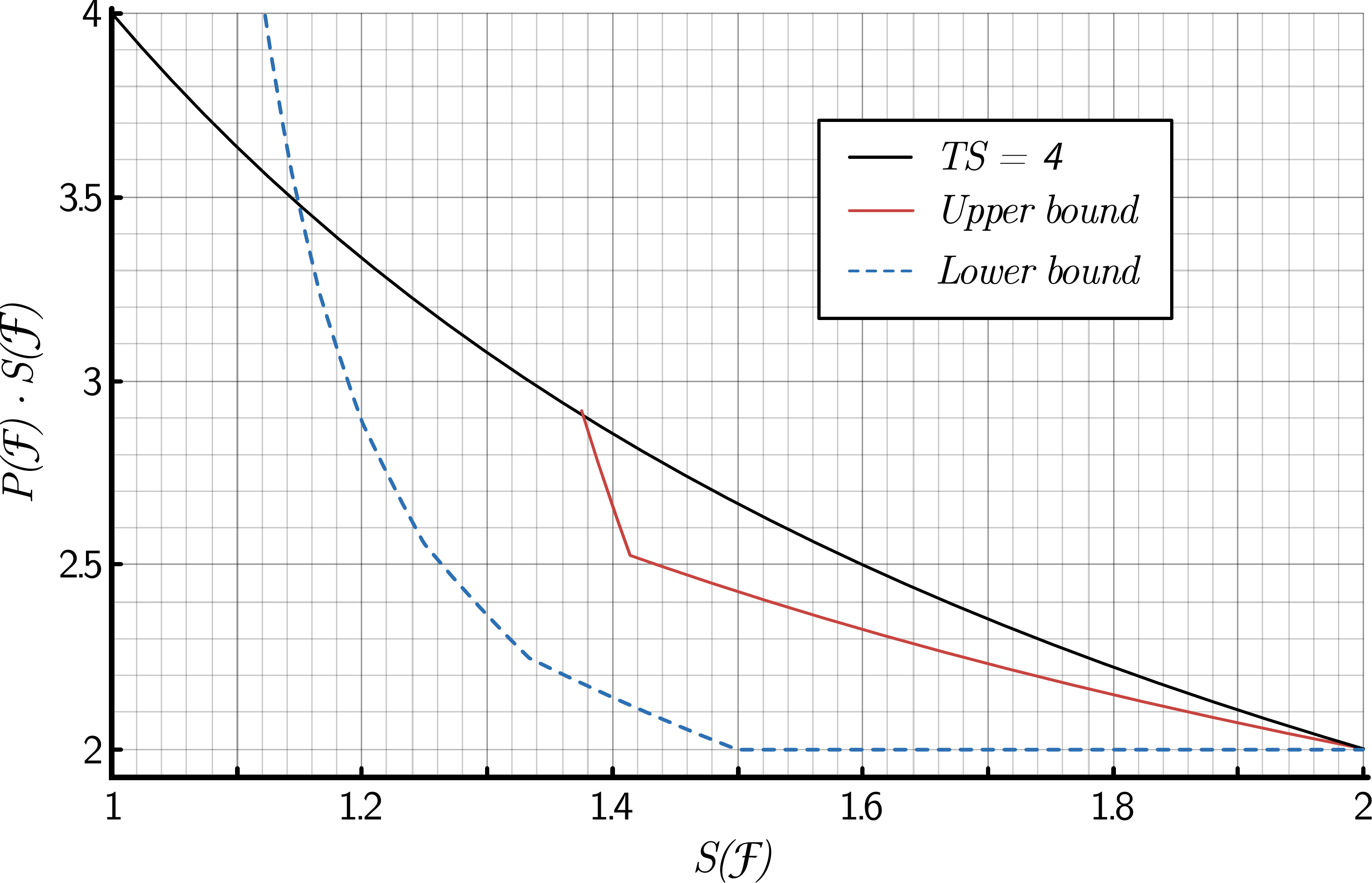}
    \caption{Possible tradeoffs between $P(\cF) \cdot S(\cF)$ and $S(\cF)$ for a single set system $\cF$. Recall that these correspond to $T$, resp.\ $S$ in the resulting space-time tradeoff. 
    Solid red line shows our upper bound, and dashed blue line is our lower bound. }
    \label{fig:LB}
\end{figure}

In this subsection we develop a lower bound on this tradeoff. Let $\cF$ be a set system over $[n]$.

\begin{lemma}\label{lem:LB}
For arbitrary $k \geq 0$ integer parameter, $P(\cF) \geq \frac{k+1}{S(\cF)^{k}} \cdot (1+o(1))$.
\end{lemma}

Here the $o$-notation is with respect to the set system ground set size $n=n(\cF)$. Note however, that a smaller value of $P(\cF)$ is not possible, even for a finite $n$, as that would imply the same value (in the limit) for $P_{S(\cF)}(n)$, by Lemma~\ref{lem:fekete}.

Before proving the lemma, let us give some interpretations. For $k=1$ it implies the lower bound $P(\cF) \cdot S(\cF) \geq 2$, meaning that the time bound $2^n$ is unavoidable. (This is intuitive, since partitioning the solution space still does not forego eventually looking at all prefix sets.)
For $k=2$ we obtain the lower bound $P(\cF) \cdot S^2(\cF) \geq 3$, showing that no point on the $ST$ tradeoff curve can improve this value (in the set systems based framework).

Figure~\ref{fig:LB} shows the obtained lower bound curve for the entire range of the parameter $k$. It is easy to see that a given value $k$ is optimal when $\frac{k+2}{k+1} \leq S(\cF) \leq \frac{k+1}{k}$.

For $k=6$ and $S(\cF) \leq (7/4)^{1/4} \approx 1.15$ we have $P(\cF) \cdot S^2(\cF) \geq 4$, and for smaller values of the normalized set system size $S(\cF)$, the product is above $4$; this means that no set system can improve the trivial tradeoff in the range $S \leq 1.15$. Recall that for this regime we combined the set system based approach with the divide and conquer algorithm.

We proceed with the proof.

\begin{proof}[Proof of Lemma~\ref{lem:LB}]
Let $\cF$ be a set system over $[n]$. We aim to give an upper bound on the number $C(\cF)$ of maximal chains supported by $\cF$.

Let $k \geq 0$, and consider the indices $i_j = \lfloor jn / (k+1) \rfloor$ for $j=0,\dots,k+1$.

Recall that a maximal chain is a sequence $S_0, \dots, S_{n} \in \cF$, where $S_{i-1} \subsetneq S_{i}$ for $i \in [n]$.

To construct a maximal chain, we first fix the entries $S_{i_j}$ for $j=0,\dots,k+1$. For $S_{i_0}$ and $S_{i_{k+1}}$ there is a single choice, $\emptyset$ and $[n]$, respectively. For each $S_{i_j}$ for $j \in [k]$ the number of choices is at most $|\cF| = S(\cF)^n$, the total number of sets in the set system. 

Having fixed $S_{i_j}$ and $S_{i_{j+1}}$, the choice for the portion of the maximal chain between them corresponds to the possible orders in which the elements $S_{i_{j+1}} \setminus S_{i_j}$ can be added, so the number of possibilities is $(i_{j+1} - i_j)! \leq (\frac{n}{k+1})!$.

The total number of maximal chains is thus $C(\cF) \leq ((\frac{n}{k+1})!)^{k+1} \cdot S(\cF)^{nk}$.

Let us finally lower bound the normalized inverse chain density $P(\cF)$.

\begin{align*}
P(\cF) & = \left( \frac{n!}{C(\cF)} \right) ^{1/n}\\
& \geq \left( \frac{n!}{((\frac{n}{k+1})!)^{k+1} \cdot S(\cF)^{nk}} \right) ^{1/n}\\
& \geq \left( \frac{(\frac{n}{e})^n}{\frac{en}{k+1} \left(\frac{n}{e(k+1)}\right)^n S(\cF)^{nk}} \right)^{1/n}\\
& \geq \frac{k+1}{S(\cF)^k} \cdot\left(\frac{k+1}{en}\right)^{1/n}\\
& \geq \frac{k+1}{S(\cF)^k} \cdot (1-o(1)).
\end{align*}

\end{proof}

Here the middle step uses the standard bounds on the factorial resulting from Stirling's approximation: $(\frac{n}{e})^n  \leq n! \leq en(\frac{n}{e})^n$.

\subsection{The Johnson-Leader-Russell conjecture}\label{sec:comb}

Johnson, Leader, and Russell~\cite{JLRsetSystems} raise the question of which set system $\cF$ over $[n]$ with a prescribed size $|\cF|$ has the largest number of maximal chains. (Or, equivalently, which set $B$ of permutations of $[n]$ of a prescribed size $|B|$ has the smallest number of prefix sets.)
The motivation of the authors in studying this question appears to be purely combinatorial, as it lies at an opposite end from Sperner-type results that allow \emph{no long chains}. Given the natural formulation, very little appears to be known about this regime, as also expressed by the authors of~\cite{JLRsetSystems}.  

Johnson, Leader, and Russell characterize the special case when $|\cF| \in \Theta(2^n)$, but leave open the setting $|\cF| \in o(2^n)$, most relevant for our study. For this case, they give a general \emph{tower of cubes} construction and conjecture that it maximizes the number of maximal chains for all appropriate set system sizes.

Precisely, for $n=tk$, a \emph{tower of $t$-cubes} can be defined as follows. Let $(P,\prec)$ be a poset with $|P| = n$, where $P$ is partitioned into antichains $P_1, \dots, P_k$, each of size $t$, where $x \prec y$ for all $x \in P_i$ and $y \in P_j$ for $i<j$.
The set system $\cF$ is then defined as the set of order ideals (downsets) of $P$. 

Johnson, Leader, and Russell conjecture~\cite[Conj.~5]{JLRsetSystems} that set systems of this form maximize the number of maximal chains among all set systems of size $|\cF| = \frac{n}{t}2^t - \frac{n}{t} +1$. In fact, they also conjecture more strongly~\cite[Conj.~6]{JLRsetSystems} that a \emph{generalized tower of cubes} construction that allows antichains to differ in size by one, is also optimal.

We first notice that the construction of Koivisto and Parviainen (see Appendix~\ref{appa}) is the special case of the tower of $t$-cubes  when $t=13$ and $k=2$.  

In~\cite{KoivistoParviainen2010} it is shown that the minimal $ST$ value is given by their $13 \times 2$ scheme, among all \emph{bucket orders}. Bucket orders are precisely the partial orders defined as a collection of antichains, linearly sorted among them, as in the tower of cubes construction above. In fact, Koivisto and Parviainen even allow for antichains of arbitrarily differing sizes. The conjecture of Johnson, Leader, and Russell would thus strongly suggest the Koivisto-Parviainen scheme to lead to an optimal time-space tradeoff, not just among poset-based approaches, but in our broader set-systems-based framework. (The reason this implication cannot be made more formally is that not all set system sizes $|\cF|$ are decomposable into cubes, and the conjecture does not cover other possible sizes.)

We argue, that this is emphatically not the case, and briefly show that the conjecture of Johnson, Leader, and Russell is false; leaving open a full characterization of extremal set systems with respect to the density of maximal chains, even at sizes that do decompose into equal cubes. 

Indeed, already for the simple case of $t = n/2$ and $k=2$, the tower of cubes construction resulting from a poset of height two (a scaled up version of the Koivisto-Parviainen construction) yields normalized size $S(\cF) = \sqrt{2} + o(1)$, and, since $C(\cF) = ((n/2)!)^2$, a normalized inverse chain density of $P(\cF) = 2-o(1)$.

This is significantly weaker than our warm-up construction (see \S\,\ref{sec2} and Appendix~\ref{appa}) with $P(\cF) \leq 1.8532$ or our best construction (\S\,\ref{sec4}) that yields $P(\cF) \leq 1.7860$  (and consequently, a number of maximal chains higher by an exponential factor) for the same set system size. In fact, the conjecture would imply that for $S(\cF) \leq \sqrt{2}$ and $n$ large enough, $S(\cF)^2 \cdot P(\cF) \geq 4$, with no set system improving the basic tradeoff. This is contradicted by our set systems obtained through the interpolation lemma, as shown on Figure~\ref{fig:LB}.

\section{Conclusion}
\label{sec5}

We significantly improved the attainable space-time tradeoff for the TSP and for more general permutation problems. Our algorithms are deterministic and arise from a new conceptual connection between dynamic programming and the existence of set systems with extremal properties. Additionally we incorporate the existing divide-and-conquer method to extend the space-time tradeoff to the entire range of parameters. While we made some effort to obtain strong numerical bounds, we also aimed to keep the constructions interpretable. Our main contribution is methodological, making the optimization task of the exponential space-time-tradeoff explicit, and connecting it with a combinatorial question of independent interest. This also allows seeing earlier approaches in a clearer framework, and we believe the approach will likely inspire results for other exponential algorithms, beyond permutation problems.

The main remaining question is closing the gap between the upper and lower bounds. While the best upper bound we obtain is $ST < 3.572$, the best lower bound we show for the set-systems-based framework is $ST \geq 3$. Notably, no better lower bound is known even when restricting attention to set systems arising as poset-ideals. Bounds that would separate this special case from general set systems would be particularly interesting. We believe that numerically improving the upper bound will likely be amenable to computer-assisted methods, with closing the gap likely requiring further ideas; we find the question of which set systems of a given size have maximal chain density to be a natural question of extremal combinatorics that deserves further investigation. Obtaining a better space-time tradeoff for TSP in the metric special case would also be interesting.

\medskip

Finally, we mention an application of our techniques to a related aspect of the complexity of TSP and related problems: their \emph{certificate}- or \emph{communication}-complexity, e.g., see~\cite{dantsin2011satisfiability}. Here, two parties (Alice and Bob) have access to the same input instance, with Alice having unlimited computational power. The tradeoff of interest is between the amount of communication from Alice to Bob, and the computation time of Bob that allows him to compute the optimum. (Notice that communication from Bob to Alice is useless; having unlimited computation, Alice can simulate Bob and anticipate his messages.)

Our set-systems-based algorithms with space $S^{n}$ and time $T^n = P^n \cdot S^{n}$ yield protocols with $\lceil n \lg{P} \rceil$ bits of communication and $S^n$ computation cost for Bob. (Alice can run the algorithm, making the step of non-deterministically ``guessing'' among $P^n$ choices for free and just communicating the choice to Bob, who can then run the rest of the algorithm.)
This should be contrasted with the trivial scheme of no communication and $\fO^{*}(2^n)$ computation time, or the folklore (best known) protocol of $ \nicefrac{n}{2} \cdot \lceil\lg{n} \rceil$ bits of communication and polynomial time computation~\cite{tsp_compl}.

%
%

\appendix

\section{Special cases}\label{appa}
We briefly discuss how some concrete results arise as special cases of the general framework. 

Every set system that admits maximal chains leads to \emph{some} space-time tradeoff, although it may not improve upon the trivial one. 
As a trivial example, consider the complete set system $\cF =  2^{[m]}$ where $S(\cF) = 2$ and $C(\cF) = 1$, with the resulting $(S,T)$ pair being the trivial $(2,2)$ of the Bellman-Held-Karp algorithm.

Another simple example over $[m]$ is a single maximal chain, where $S(\cF) = (m+1)^{1/m}$ and $P(\cF) = (m!)^{1/m}$. The tradeoff resulting for this set system is at the extreme end of saving space. E.g., for $m = 8$ it results in an algorithm with space $S \leq 1.3161$ and time $T \leq 4.9552$, significantly above the Gurevich-Shelah bounds for both parameters. 

\medskip
The warmup example of \S\,\ref{sec2} can be reproduced via Theorem~\ref{thm:main} and the following set system.

Let $\beta \approx 0.889972$, the root of $H(x) = \nicefrac{1}{2}$, and define $\cF$ as a set system over $[2k]$, where $\cF = \cF' \cup \cF'' \cup \cF'''$, with the following definitions: $\cF' = 2^{[k]}$, $\cF'' = \{(s+k) \cup [k] \mid s \in 2^{[k]}\}$, and $\cF''' = \{(s_1 \cup (s_2 + k))\mid  s_1, s_2 \in 2^{[k]}, |s_1| \geq \beta k, |s_2| \leq (1-\beta)k\}$.

For the relevant parameters of the set system, the same calculations as in \S\,\ref{sec2} yield $S(\cF) = 2^{1/2} \cdot (1+o(1))$ and $P(\cF) \leq 1.8532$ (for large enough $k$). (For the latter, we can look at the probability that a random permutation $\pi$ of $[2k]$ is supported by $\cF$, which only depends on the half-prefix-set $\pi^{(k)}$.)

\medskip

A natural way to construct set systems is by taking the order ideals of a partial order (poset). 

More precisely, for a poset $(P, \prec)$, the set system of its order ideals (downsets) is the collection of sets $\cF = \{s \subseteq P \mid \forall x \in s, \forall y \in P, (y \prec x \implies y \in s) \}$.

As the number of set systems over $[n]$ is $2^{2^{[n]}}$ and the number of posets is $2^{\fO(n^2)}$~\cite{kleitman1975asymptotic}, it is to be expected that set systems offer more flexibility. Order ideals of a poset are closed under union and intersection, which need not be the case for general set systems; for example, if the set of order ideals of a poset contains the singleton sets $\{1\}, \dots, \{n\}$, then it necessarily contains all subsets of $[n]$. 

The set system capturing the best bound in the framework of Koivisto and Parviainen arises from a height-two poset $([26],\prec)$ and can be seen as follows. 

Let $\cF$ be over $[26]$, with $\cF = 2^{[13]} \cup ´\{(s+13) \cup [13] \mid s \in 2^{[13]}\}$.

It is easy to see that $S(\cF) = (2^{13} + 2^{13} - 1)^{1/26} \approx 1.453$, and $P(\cF) = \left(\frac{26!}{ (13!)^2  }\right)^{1/26} \approx 1.862$, resulting in the space-time tradeoff $ST \leq P(\cF) \cdot S(\cF)^2 \approx 3.93$. 

\newpage

\small
\bibliographystyle{alphaurl}
\bibliography{main}
\end{document}